\documentclass{article}
\usepackage{amsmath}
\usepackage[T1]{fontenc}
\usepackage[utf8]{inputenc}
\usepackage{amsthm}
\usepackage{amstext}
\usepackage{amssymb}
\usepackage{hyperref}
\usepackage[langfont=caps]{complexity}
\usepackage{authblk}
\usepackage{bbm} % Fonction indicatrice
\usepackage{mathrsfs} % Lettres caligraphiques (\mathscr)

\newtheorem{theorem}{Theorem}

\newtheorem{lemma}[theorem]{Lemma}
\newtheorem*{lemma*}{Lemma}
\newtheorem{proposition}[theorem]{Proposition}
\newtheorem{corollary}[theorem]{Corollary}
\newtheorem{remark}[theorem]{Remark}

\newtheorem*{open*}{Open~question}
\newtheorem{definition}[theorem]{Definition}

\renewcommand\cc{\ensuremath{\mathbb{C}}}
\newcommand\rr{\ensuremath{\mathbb{R}}}
\newcommand\nn{\ensuremath{\mathbb{N}}}

\begin{document}

\title{On the intersection of a sparse curve\\ and a low-degree curve:\\
A polynomial version of the lost theorem}

\author{Pascal Koiran, Natacha Portier, and S\'ebastien Tavenas}

\affil{LIP\footnote{UMR 5668 ENS Lyon - CNRS - UCBL - INRIA,
    Universit\'e de Lyon. Email:
    \{pascal.koiran,natacha.portier,sebastien.tavenas\}@ens-lyon.fr. 
The authors are supported by ANR project CompA (project number: ANR--13--BS02--0001--01).}, \'Ecole Normale Sup\'erieure de Lyon}

\date{\today}

\maketitle

\begin{abstract}
Consider a system of two polynomial equations in two variables: 
$$F(X,Y)=G(X,Y)=0$$
 where $F \in \rr[X,Y]$ has degree $d \geq 1$ and $G \in \rr[X,Y]$ 
has $t$ monomials.
We show that the system has only $O(d^3t+d^2t^3)$ real solutions 
when it has a finite number of real solutions.
This is the first polynomial bound for this problem.
In particular, the bounds coming from the theory of fewnomials are exponential in $t$, and count only nondegenerate solutions.
More generally, we show that if the set of solutions is infinite,
it still has at most $O(d^3t+d^2t^3)$ connected components.

By contrast, the following question seems to be open: if $F$ and $G$
have at most $t$ monomials, is the number of (nondegenerate) solutions 
polynomial in $t$?

The authors' interest for these problems was sparked by connections 
between lower bounds in algebraic complexity theory and upper bounds on
the number of real roots of ``sparse like'' 
polynomials.
\end{abstract}

\section{Introduction}

Descartes' rule of signs shows that a real univariate polynomial 
with $t \geq 1$ monomials has at most $t-1$ positive roots.
%A polynomial with at most $t$ monomials is often called ``$t$-sparse''.
In 1980, A. Khovanskii~\cite{Kho91} obtained a far reaching
generalization. He showed that a system of $n$
polynomials in $n$ variables involving $l+n+1$ distinct monomials has
less than 
\begin{align} \label{khovanskii}
  2^{\binom {l+n}2}(n+1)^{l+n}
\end{align} 
non-degenerate positive solutions. 
Like Descartes', this bounds  depends on the number of
monomials of the polynomials but not on their degrees.

In his theory of fewnomials (a term coined by Kushnirenko),
Khovanskii~\cite{Kho91} gives
a number of results of the same flavor; some apply to non-polynomial functions.
In the case of polynomials, Khovanskii's result was improved by
Bihan and Sottile~\cite{BS07}. Their bound is
\begin{align} \label{biso}
  \frac{e^2+3}{4}2^{\binom l2}n^l.
\end{align}

%These two bounds suppose that each one of the $l+1+n$ monomials can appear
%in each polynomial. Consequently, lower
%bounds could be attained
%if some monomials do not appear in some polynomials. 
%{\bf (pr\'eciser ou supprimer)}

In this paper, we bound the number of real solutions of a system 
\begin{align}      \label{system}
  F(X,Y)=G(X,Y)=0
\end{align}
of two polynomial equations in two variables,
where  $F$ is a polynomial of degree  $d$ and $G$ has $t$ monomials. 
This problem has a peculiar history~\cite{BS11,Kush08,Sot11}.
Sevostyanov showed
in 1978   that %when the number of solutions is finite,  it
the number of nondegenerate solutions can be bounded by a  function
$N(d,t)$ which depends on $d$ and $t$ only.
According to~\cite{Sot11}, this result was the inspiration for Khovanskii to develop his theory of fewnomials.
Sevostyanov suffered an early death, and his result was never published.
Today, it seems that Sevostyanov's proof and even the specific form of
his  bound have been lost.

The results of Khovanskii~(\ref{khovanskii}), 
or of Bihan and Sottile~(\ref{biso}), imply a bound on $N(d,t)$
which is exponential in $d$ and $t$. 
Khovanskii's bound~(\ref{khovanskii}) follows from a general result 
on mixed polynomial-exponential systems (see Section~1.2 of~\cite{Kho91}).
One can check that the latter result implies a bound on $N(d,t)$ which
is exponential in $t$ only.
As we shall see, this is still far from optimal.

Li, Rojas and Wang~\cite{LRW03} showed that the number of real roots
is bounded above by  $2^t-2$ when $F$ is a trinomial.
When  $F$ is linear, this bound was improved to  $6t-4$ by
Avenda\~{n}o~\cite{Ave09}. 
The result by  Li, Rojas and Wang~\cite{LRW03} is in fact more general: 
they show that the number of
non-degenerate positive real solutions of the system 
\begin{align*}
  F_1(X_1,\ldots,X_n)=F_2(X_1,\ldots,X_n)=\ldots=F_n(X_1,\ldots,X_n)=0
\end{align*}
is at most $n + n^2 +\ldots + n^{t-1}$ when each of
$F_1,\ldots,F_{n-1}$ is a trinomial and $F_n$ has $t$ terms.

Returning to the case of a system $F(X,Y)=G(X,Y)=0$ where $F$ is a trinomial
and $G$ has $t$ terms, we obtained in~\cite{KPT13} a $O(t^3)$ upper bound
on the number of real roots. It is also worth pointing out that, 
contrary to~\cite{Ave09}, the methods of~\cite{KPT13} apply to systems with
real exponents.

The present paper deals with the general case of Sevostyanov's
system~(\ref{system}). 
We obtain the first bound which is polynomial in $d$ and $t$.
Indeed, we show that there are only $O(d^3t+d^2t^3)$ real solutions
to~(\ref{system}) when their number is finite. 
Note that we count all roots, including degenerate roots.
More generally, we show that when the set of solutions is infinite 
the same $O(d^3t+d^2t^3)$ upper bound  applies  to the number
of its connected components (but it is actually the finite case which requires
most of the work).

Note finally that our bound  applies only when $F$ is a polynomial
of degree $d \geq 1$.
As pointed out in Section~\ref{main}, the case $d=0$ is more difficult.
The reason is that a system of two sparse equations 
can be encoded in a system where $F=0$.
We do not know if the number of real roots can be bounded by a polynomial
function of~$t$ in this case.

The authors' interest for these problems was sparked by connections 
between lower bounds in algebraic complexity theory and upper bounds on
the number of real roots of ``sparse like'' 
polynomials: see~\cite{Koi10a,GKPS11,KPT13} 
as well as the earlier work~\cite{BC76,Gri82,Risler85}.

\subsection*{Overview of the proof}

As we build on results from~\cite{KPT13}, it is helpful to recall how
the case $d=1$ (intersection of a sparse curve with a line) was
treated in that paper. For a line of equation $Y=aX+b$, this amounts
to bounding the number of real roots of a univariate polynomial of the
form
\begin{align*}
  \sum_{i=1}^t c_i X^{\alpha_i} (aX+b)^{\beta_i}.
\end{align*}
This polynomial is presented as a sum of $t$ ``basis functions'' of
the form $f_i(X)=c_i X^{\alpha_i} (aX+b)^{\beta_i}.$ In order to bound
the number of roots of a sum of real analytic functions, it suffices
to bound the number of roots of their Wronskians. We recall that the
Wronskian of a family of functions $f_1,\ldots,f_k$ which are $(k-1)$
times differentiable is the determinant
of the matrix of their derivatives of order 0 up to $k-1$. More
formally,
\begin{align*}
  W(f_1,\ldots,f_k) = \det \left(\left( f_j^{(i-1)}\right)_{1\leq i,j
      \leq k}\right).
\end{align*}
In~\cite{KPT13}, we proved the following result.

\begin{theorem} \label{thm_Wronskien}
  Let $I$ be an open interval of $\rr$ 
and let $f_1, \ldots, f_t:I \rightarrow \rr$ be a family of analytic 
functions which are linearly independent on $I$. 
For $1\leq i\leq t$, let us denote by $W_i:I \rightarrow \rr$ 
the Wronskian of $f_1,\ldots,f_i$.
Then, 
  \begin{align*}
    Z(f_1+\ldots + f_t) \leq t-1 + Z(W_t)+Z(W_{t-1})+2\sum_{j=1}^{t-2}Z(W_j)
  \end{align*}
  where $Z(g)$ denotes the number of distinct real roots of 
a function $g:I \rightarrow \rr$.
\end{theorem}

The present paper again relies on Theorem~\ref{thm_Wronskien}.
Let us assume that for a system $F(X,Y)=G(X,Y)=0$, 
we can use the equation $F(X,Y)=0$ to express $Y$ as an (algebraic) function
of $X$. Then we just have to bound the number of real roots of a
univariate polynomial of the form 
$$\sum_{i=1}^t c_i X^{\alpha_i} \phi(X)^{\beta_i},$$
and this is a situation where we can apply Theorem~\ref{thm_Wronskien}.
Of course, turning this informal idea into an actual  proof requires some
care. In particular, the algebraic function $\phi$ needs not be defined
on the whole real line, and it needs not be uniquely defined.
We deal with those issues using Collin's cylindrical algebraic decomposition
(see Section~\ref{cad}).
We also need some quantitative estimates on 
the higher-order derivatives of the algebraic function $\phi$ because
they appear in the Wronskians of Theorem~\ref{thm_Wronskien}.
For this reason, we express in Section~\ref{algder} 
the derivatives of $\phi$ in terms of $\phi$ and 
of the partial derivatives of $F$.
Using Theorem~\ref{thm_Wronskien}, we can ultimately reduce Sevostyanov's 
problem to the case of a system where both polynomials have bounded degree.
The relevant bounds for this case are recorded in Section~\ref{real}.
We put these ingredients together in Section~\ref{main}
to obtain the $O(d^3t+d^2t^3)$ bound on the number of connected components.

\section{Technical Tools}

In this section we collect various  results that are required 
for the main part of this paper (Section~\ref{main}).
On first reading, there is no harm in beginning with Section~\ref{main};
the present section can be consulted when the need arises.

\subsection{The derivatives of a power}

In this section, we recall how the derivatives of a power of a
univariate function~$f$ can be expressed in terms of the derivatives
of $f$.  We use ultimately vanishing sequences of integer numbers,
i.e., infinite sequences of integers which have only finitely many
nonzero elements. We denote the set of such sequences
$\nn^{(\nn)}$. For any positive integer $p$, let $\mathscr{S}_p = \{
(s_1,s_2,\ldots) \in \nn^{(\nn)} | \
\underset{i=1}{\overset{\infty}{\sum}} i s_i = p \}$ (so in particular
for each $p$, this set is finite). Then if $s$ is in $\mathscr{S}_p$,
we observe that for all $i\geq p+1$, we have $s_i=0$. Moreover for any
$p$ and any $s=(s_1,s_2,\ldots) \in \nn^{(\nn)}$, we will denote $|s|
= \underset{i=1}{\overset{\infty}{\sum}} s_i$ (the sum makes sense
because it is finite).  A proof of the following simple lemma can be
found in~\cite{KPT13}.

\begin{lemma} \label{lem_power}[Lemma 10 in~\cite{KPT13}] Let $p$ be a
  positive integer.  Let $f$ be a real function and $\alpha \geq p$ be
  a real number such that $f$ is always non-negative or $\alpha$ is an
  integer (this ensures that the function $f^{\alpha}$ is well defined).
  Then \[\left(f^\alpha \right)^{(p)} = \sum_{s \in \mathscr{S}_{p}}
  \left[ \beta_{\alpha,s} f^{\alpha-|s|} \prod_{k=1}^p \left( f^{(k)}
    \right)^{s_k} \right] \] where $(\beta_{\alpha,s})$ are some
  constants.
\end{lemma}

The order of differentiation of a monomial $\prod_{k=1}^p
(f^{(k)})^{s_k}$ is $\sum_{k=1}^p ks_k$. The order of differentiation
of a differential polynomial is the maximal order of its
monomials. For example: if $f$ is a function, the total order of
differentiation of
$f^3\left(f^\prime\right)^2\left(f^{(4)}\right)^3+3ff^\prime$ is $\max
(3*0+2*1+3*4,0*1+1*1)=14$.

Lemma~\ref{lem_power} just means that the $p$-th derivative of an
$\alpha$th power of a function $f$ is a linear combination of
terms such that each term is a product of derivatives of $f$ of total
degree $\alpha$ and of total order of differentiation $p$.

\subsection{The derivatives of an algebraic function} \label{algder}

Consider a nonzero bivariate polynomial $F(X,Y) \in \rr[X,Y]$ and a
point $(x_0,y_0)$ where $F(x_0,y_0)=0$ and the partial derivative
$F_Y=\frac{\partial F}{\partial Y}$ does not vanish.  By the implicit
function theorem, in a neighborhood of $(x_0,y_0)$, the equation
$F(x,y)=0$ is equivalent to a condition of the form $y=\phi(x)$. The
implicit function $\phi$ is defined on an open interval $I$ containing
$x_0$, and is $C^{\infty}$ (and even analytic).  In this section, we
express the derivatives of $\phi$ in terms of $\phi$ and of the
partial derivates of $F$.  For any integers $a, b$, we denote
$F_{X^aY^b}=\frac{\partial^{a+b}}{\partial X^a\partial Y^b} F(X,Y)$.

\begin{lemma}\label{Lem_higherDerivative}
  For all $k\geq 1$, there exists a polynomial $S_k$ 
  of degree at most $2k-1$ in $\binom{k+2}2-1$
  variables such that 
  \begin{align}\label{Eq_RationnalFunction}
    \phi^{(k)}(x) =
    \frac{S_k\left(F_X(x,\phi(x)),\ldots,F_{X^aY^b}(x,\phi(x)),\ldots\right)}{(F_Y(x,\phi(x)))^{2k-1}}
  \end{align}
  with $1\leq a+b\leq k$. Consequently, the numerator is a polynomial
  of total degree at most
  $(2k-1)d$ in $x$ and $\phi(x)$. Moreover, $S_k$ depends only on $k$
  and $F$.
\end{lemma}

\begin{proof} For all $k$, let $D_k(x) = \frac{\partial^k}{\partial
    x^k}F(x,\phi(x))$. We will use later the fact that $D_k(x)$ is the
  identically zero function. We begin by showing by induction that for
  all $k\geq1$,
  $D_k(x)=\phi^{(k)}F_Y+R_k(\phi^\prime(x),\ldots,\phi^{(k-1)}(x),\ldots,F_{X^aY^b},\ldots)$
  where $R_k$ is of total degree at most $1$ in $(F_{X^aY^b})_{1\leq
    a+b\leq k}$ and of derivation order at most $k$ in the variables
  $\left(\phi^{(i)}\right)_{1\leq i<k}$.

  For $k=1$, we get: $D_1=\phi^\prime F_Y+F_X$.  Let us suppose now
  that the result is true for a particular $k$, then
  \begin{align*}
    D_{k+1}= &
    \phi^{(k+1)}F_Y+\phi^{(k)}(F_{XY}+F_{Y^2}\phi^\prime)+\frac\partial{\partial
      x}R_k\left(\phi^\prime,\ldots,\phi^{(k-1)},F_{X^aY^b}\right) \\
    =&
    \phi^{(k+1)}F_Y+R_{k+1}\left(\phi^\prime,\ldots,\phi^{(k)},F_{X^aY^b}\right).
  \end{align*}
  By induction hypothesis, each monomial of $R_k$ is of the form 
  \begin{align*}
    F_{X^aY^b}\prod_{i=1}^{k-1} \left(\phi^{(i)}\right)^{s_i}
    \textrm{ where }\sum_{i=1}^{k-1} is_i \leq k.
  \end{align*}
  Differentiating this monomial increases its order of differentiation
  by one at most.
  Indeed,
  \begin{align*}
    & \frac{\partial}{\partial x} \left(F_{X^aY^b}\prod_{i=1}^{k-1}
      \left(\phi^{(i)}\right)^{s_i} \right)  \\
    & = (F_{X^{a+1}Y^b}+\phi^\prime F_{X^aY^{b+1}}) \prod_{i=1}^{k-1}
    \left(\phi^{(i)}\right)^{s_i} \\
    & \quad\quad + F_{X^aY^b}\sum_{i=1}^{k-1} s_i \phi^{(i+1)}
    \left(\phi^{(i)}\right)^{s_i-1} \prod_{j\neq i} \left(\phi^{(j)}
    \right)^{s_j}.
  \end{align*}
  Hence, $R_{k+1}$ is of total degree at most $1$ in the variables
  $(F_{X^aY^b})_{1\leq a+b\leq k+1}$ and of derivation order at most
  $k+1$ in $\left(\phi^{(i)}\right)_{1\leq i\leq k}$.

  As $F(x,\phi(x))$ is zero, then for all $k\geq1$ we have
  $D_k(x)=\frac{\partial^k F(x,\phi(x))}{\partial x^k} =0$. Thus
  \begin{align*}
    \phi^{(k)}= \frac{-R_k}{F_Y}.
  \end{align*}
  Then we show by induction over $k$ that for all $k\geq 1$ there
  exists a polynomial $S_k$ of degree at most $2k-1$ in $\left(\binom{k+2}2-1\right)$ variables
  such that Equation~(\ref{Eq_RationnalFunction}) is verified.
  
  The result is true for $k=1$ since
  $\phi^\prime=\frac{-F_X}{F_Y}$. Let $k\geq1$ and we suppose that the result is
  true for all $i$ such that $1\leq i\leq k$. We know that
  $D_{k+1}(x)=\phi^{(k+1)}F_Y+R_{k+1}(\phi^\prime(x),\ldots,\phi^{(k)}(x),\ldots,F_{X^aY^b},\ldots)=0$. So,
  \begin{align*}
    \phi^{(k+1)}=\frac{-1}{F_Y}R_{k+1}(\phi^\prime(x),\ldots,\phi^{(k)}(x),\ldots,F_{X^aY^b},\ldots).
  \end{align*}
  So, by induction hypothesis,
  \begin{align*}
    \phi^{(k+1)}=\frac{-1}{F_Y}R_{k+1}\left(\frac{S_1}{F_Y},\ldots,\frac{S_k}{F_Y^{2k-1}},\ldots,F_{X^aY^b},\ldots\right).
  \end{align*}
  As $R_{k+1}(\phi^\prime(x),\ldots,\phi^{(k)}(x),\ldots,F_{X^aY^b},\ldots)$ is of derivation order $k+1$ on its $k$ first variables
  and is of total order $1$ on its $\left(\binom{k+2}2-1\right)$ last
  variables, each monomial is of the form:
  \begin{align*}
    F_{X^aY^b}\frac{S_{i_1}}{F_Y^{2i_1-1}}\ldots \frac{S_{i_p}}{F_Y^{2i_p-1}}
  \end{align*} with $i_1+\ldots+i_p\leq k+1$ and $p\geq 0$.
%  Here, there are two cases.
%  \begin{itemize}
%  \item Either $p=0$ and in this case, the monomial is of the form 
%  \begin{align*}
%    \frac{F_{X^aY^b} F_Y^{2(k+1)-2}}{F_Y^{2(k+1)-2}}
%  \end{align*} where the numerator is a polynomial on the variables $F_{X^aY^b}$
%    of degree $\leq 1+2(k+1)-2\leq 2(k+1)-1$.
%  \item Or $p\geq 1$, and in this case, the monomial is of the form 
  Hence, we get:    
  \begin{align*}
    \frac{F_{X^aY^b} S_{i_1}\ldots S_{i_p}}{F_Y^{2i_1-1}\ldots F_Y^{2i_p-1}}=
    \frac{F_{X^aY^b} S_{i_1}\ldots S_{i_p}F_Y^{2k-2i+p}}{F_Y^{2(k+1)-2}}
  \end{align*} where $i=i_1+\ldots+i_p\leq k+1$. Indeed, the exponent
  $2k-2i+p$ is a non-negative integer since if $p=1$, then
  $2i=2i_1\leq 2k$ and otherwise $2i\leq 2(k+1) \leq 2k+p$. The numerator is a polynomial in the variables $F_{X^aY^b}$
  of degree 
  \begin{align*}
    & \leq 1+\deg(S_{i_1})+\ldots+\deg(S_{i_p})+2k-2i+p\\
    & \leq 1+2i_1-1+\ldots+2i_p-1+2k-2i+p\\
    & \leq 1+2i-p+2k-2i+p\\
    & \leq 2(k+1)-1.
  \end{align*}
%  \end{itemize}
  So, $\phi^{(k+1)}$ is of the form:
  \begin{align*}
    \frac{S_{k+1}\left(\left(F_{X^aY^b}\right)_{1\leq a+b\leq
          k+1}\right)}{F_Y^{2(k+1)-1}}
  \end{align*} where $S_{k+1}$ is a polynomial of degree at most $2(k+1)-1$.
\end{proof}

\subsection{Real versions of B\'ezout's theorem} \label{real}

B\'ezout's theorem is a fundamental result in algebraic geometry.
%concerning the
%number of common points, or intersection points, of $n$ plane
%algebraic curves on an algebraically closed field. 
One version of it is as follows.
\begin{theorem}\label{Thm_Bezout}
Consider an algebraically closed field $K$ and $n$ polynomials 
$f_1,\ldots,f_n \in K[X_1,\ldots,X_n]$ of degrees $d_1,\ldots,d_n$. 
If the polynomial system
$$f_1=f_2=\cdots=f_n=0$$
has a finite number of solutions in $K^n$, this number is at most $\displaystyle \prod_{i=1}^n d_i$.
\end{theorem}
The upper bound $\prod_{i=1}^n d_i$ may not apply if $K$ is not
algebraically closed. In particular, it fails for the field of real
numbers (see e.g. chapter~16 of~\cite{BCSS98} for a counterexample).
Nevertheless, there is a large body of work establishing bounds of a
similar flavor for $K=\rr$ (see e.g.~\cite{BPR06,BCSS98} and the
references therein).  For instance, we have the following classic
result.
%the following result is established in chapter~16 
%of~\cite{BCSS98}.
\begin{theorem}[Oleinik-Petrovski-Thom-Milnor] \label{BCSSbezout} Let
  $V \subseteq \rr^n$ be defined by a system $f_1 = 0,\ldots, f_p =
  0$, where the $f_i$ are real polynomials of degree at most $d$
  with $d\geq 1$.  Then the number of connected components of $V$
  is at most $d(2d-1)^{n-1}$.
\end{theorem}
A proof of Theorem~\ref{BCSSbezout} can be found in e.g. Chapter~16 
of~\cite{BCSS98}.
In this paper we will use this result as well as 
 a variation for the case $n=p=2$
(see Lemma~\ref{Lem_realBezout} at the end of this section).
Lemma~\ref{Lem_decompose} below will be also useful
in Section~\ref{main}. %the proof of Theorem~\ref{Cor_genSystem}.
We now give self-contained proofs of these two lemmas 
since they are quite short.

\begin{lemma}\label{Lem_decompose}
Let $g \in \rr[X,Y]$  be a non-zero polynomial of degree $d$.
The set of real zeros of $g$ is the union of a set of at most $d^2/4$ points and of the zero sets of polynomials $g_1,\ldots,g_k \in \rr[X,Y]$
which divide $g$ and are irreducible in $\cc[X,Y]$.
\end{lemma}
\begin{proof}
Let us factor $g$ as a product of irreducible polynomials in $\cc[X,Y]$.
We have 
$$g=\lambda g_1^{\alpha_1} \cdots g_k^{\alpha_k} h_1^{\beta_1} \cdots h_l^{\beta_l}  \overline{ h_1}^{\beta_1} \cdots \overline{h_l}^{\beta_l}$$
where the $g_j$ are the factors in $\rr[X,Y]$, the polynomials
$h_j,\overline{h_j}$ are complex conjugate and $\lambda$ is  a real
constant. 
We can assume that none of the $h_j$ is of the form
$h_j = \mu_j r_j$ where $\mu_j \in \cc$ and $r_j \in \rr[X,Y]$:
otherwise, we can replace the pair $(h_j,\overline{h_j})$ 
by $r_j^2$ and the constant $\mu_j \overline{\mu_j}$ can be absorbed by
$\lambda$.

The above assumption implies that the $h_j$ (and their conjugates) 
have finitely many real zeros. 
Indeed,  let $p_j, q_j$ be the real and imaginary parts of $h_j$.
The real solutions of $h_j=0$ are the same as those of $p_j=q_j=0$. This system has finitely many complex solutions since $p_j$ and $q_j$ are nonzero and do not share a common factor.
%(otherwise, $h_j=p_j+iq_j$ would not be irreducible).
Consider indeed a putative factor $f_j \in \cc[X,Y]$ dividing $p_j$ and $q_j$,
of degree $\deg(f_j) \geq 1$. Since $f_j$ divides $h_j$ 
and this polynomial is irreducible, we must have $\deg(f_j)=\deg(h_j)$.
As a result, $\deg(p_j)=\deg(q_j)=\deg(f_j)$ and the first two polynomials
are constant multiples of the third. 
We conclude that $p_j$ differs from $q_j$ only by a multiplicative 
constant, and this contradicts our assumption.

By B\'ezout's theorem, there are at most $\deg(h_j)^2$ complex solutions to $p_j=q_j=0$. This is also an upper bound on 
the number of real roots of the $h_j$. The $\overline{h_j}$ have the same real roots.
Altogether, the $h_j$ and $\overline{h_j}$ have at most $\sum_{j=1}^l \deg(h_j)^2 \leq (d/2)^2$ real roots.
\end{proof}
The point of this lemma is that since each $g_j$ is
irreducible, the set of its singular zeros (i.e.,  the set of complex
solutions of the system $g={\partial g / \partial x} = {\partial g} /
{\partial y} = 0$) is finite and small. We first consider the more
general case given by the system $g={\partial g / \partial x} = 0$.
\begin{lemma}
If $\ g \in \cc[X,Y]$ is an irreducible polynomial of degree $d
\geq 1$, then either $g(X,Y)$ is of the form $aY+b$ or the number of
zeros in $\cc^2$ of $g={\partial g / \partial x} = 0$ is at most $d(d-1)$.
\end{lemma}
\begin{proof}
  We consider two cases.
  \begin{itemize}
  \item[(i)] If the system $g=\partial g / \partial x = 0$ has finitely many solutions, it has at most $d(d-1)$ solutions by B\'ezout's theorem. 
  \item[(ii)] If that system has infinitely many solutions, $\partial g / \partial x$ must vanish everywhere on the zero set of $g$ since $g$ is irreducible. For the same reason, it then follows that $g$ divides $\partial g / \partial x$. This is impossible by degree considerations unless $\partial g / \partial x \equiv 0$ on $\cc^2$. Hence $g$ depends only on the variable $Y$, and must be of the form $g(X,Y)=aY+b$ (by irreducibility again). 
  \end{itemize}
\end{proof}
As the additional condition ${\partial g/\partial y}=0$ implies $a=0$
in the previous lemma, we have:
\begin{corollary}\label{Cor_realBezout}
  If $\ g \in \cc[X,Y]$ is an irreducible polynomial of degree $d \geq 1$, it has at most $d(d-1)$ singular zeros in $\cc^2$.
\end{corollary}

We are now going to bound the number of roots of a real system of two dense
equations.
\begin{lemma}\label{Lem_realBezout}
  Let $f, g \in \rr[X,Y]$  be two non-zero polynomials of respective degrees 
$\delta$ and $d$. Let $\mathcal{U}$ be an open subset of $\rr^2$. 
Consider the system of polynomial equations:
  \begin{align}\label{Eq_realBezout}
    \begin{cases}
      f(X,Y)=0 \\ g(X,Y)=0.
    \end{cases}
  \end{align} 
If the number of solutions in $\mathcal{U}$ is finite, %this number
it is bounded by
  $d^2/4+d\delta$.

Moreover, if $f$ is the zero polynomial, the number of solutions in 
$\mathcal{U}$
 of the same system  is infinite or bounded by $d^2$.
\end{lemma}

\begin{proof}
  Let us suppose that the system has finitely many solutions in
  $\mathcal{U}$.  By Lemma~\ref{Lem_decompose}, the set of roots of
  $g$ is the union of a set of size at most $d^2/4$ and of the sets of
  roots of the real polynomials $g_1,\ldots,g_k$.  Hence the number of
  solutions of System~(\ref{Eq_realBezout}) is bounded by $d^2/4$ plus
  the sum of the numbers of solutions of each system
  $g_i(X,Y)=f(X,Y)=0$.  Let us define $d_i$ as the degree of $g_i$.
  For each $i$, there are two cases:
  \begin{itemize}
  \item[(i)] if $g_i$ divides $f$ then either the number of real roots
    of $g_i$ is infinite on $\mathcal{U}$ and then all these roots are
    solutions of System~(\ref{Eq_realBezout}) or the number of its
    roots is finite and in this case, each of these roots is a
    singular zero of $g_i$ (if a point is an isolated zero of a
    continuous function on an open set of the plane, it needs to be an
    extremum of the function on this set). 
Hence by Corollary \ref{Cor_realBezout},
    the number of real roots is bounded by $d_i(d_i-1)$, and so, since
    $g_i$ divides $f$, by $d_i\delta$ if $f$ is not zero.
  \item[(ii)] otherwise, $g_i$ does not divide $f$ and thus the system has
    a finite number of solutions in $\mathbb{C}^2$ and this number is
    bounded by $d_i\delta$ according to Bézout's theorem.
  \end{itemize}
  Thus for each $i$, the number of solutions of the system
  $g_i(X,Y)=f(X,Y)=0$ is at most $d_i\delta$.

  In the case where $f$ is the zero polynomial, we can argue as in
  case (i). We saw in the proof of Lemma~\ref{Lem_decompose} that the
zero set of $g$ is the union of the zero sets of the $g_i$ and of the
$h_i$, and that altogether the $h_i$ have at most $\sum_i \deg(h_i)^2$
zeros. Moreover, as in case (i), the number of zeros of $g_i$ on $\mathcal{U}$
is infinite or bounded by $d_i(d_i-1)$.
We conclude that the number of zeros of $g$ on $\mathcal{U}$ is
infinite or bounded by 
$$\sum_i \deg(h_i)^2 + \sum_i d_i(d_i-1) \leq \sum_i \deg(h_i)^2 +
\sum_i d_i^2 \leq d^2.$$
\end{proof}
Note that  the somewhat worse bound  $$\max(d,\delta).(2\max(d,\delta)-1)$$
follows directly from Theorem~\ref{BCSSbezout}.

\subsection{Cylindrical algebraic decomposition for one bivariate polynomial} \label{cad}

In his paper, Collins~\cite{Col75} introduced the cylindrical algebraic
decomposition. The purpose was to get an algorithmic proof of quantifier
elimination for real closed fields. More details on cylindrical
algebraic decomposition can be found in~\cite{BPR06}. Here, we will use a similar
decomposition of $\rr$ for separating the different behaviours of the
roots of our system. However, in our case the dimension is just two,
and we want to characterize only one polynomial, so we will use an easier
decomposition. Let us recall some definitions and properties 
%which come 
from~\cite{Col75}.

\begin{definition}
  Let $A(X,Y)$ be a real polynomial, $S$ a subset of $\rr$.
  We will say that $f_1,\ldots,f_m$ with $m\geq 1$ delineate the roots
  of $A$ on $S$ in case the following conditions are all satisfied:
  \begin{enumerate}
  \item  $f_1,\ldots,f_m$ are distinct continuous functions from $S$ to $\cc$.
  \item For all $1\leq i\leq m$, there is a positive integer $e_i$ such
    that for all $a$ in $S$, $f_i(a)$ is a root
    of $A(a,Y)$ of multiplicity $e_i$.
  \item If $a\in S$, $b\in\cc$ and $A(a,b) = 0$ then for some $i$ with
    $1\leq i\leq m$, $b=f_i (a)$.
  \item For some $k$ with $0\leq k\leq m$, the functions
    $f_1,\ldots,f_k$ are real-valued with $f_1<f_2<\dotsb<f_k$ and the
    values of $f_{k+1},\ldots,f_m$ are all non-real.
  \end{enumerate}
  The value $e_i$ will be called the multiplicity of $f_i$. If $k\geq
  1$, we will say that $f_1,\ldots,f_k$ delineate the real roots of
  $A$ on $S$. The roots of $A$ are delineable on $S$ in case there are
  functions $f_1,\ldots,f_m$ which delineate the roots of $A$ on $S$.
\end{definition}

Collins proved the following theorem.
\begin{theorem}[Particular case of Theorem 1 in~\cite{Col75}]
  Let $A(X,Y)$ be a polynomial in $\rr[X,Y]$. Let $S$ be a connected
  subset of $\rr$. If the leading coefficient of $A$ viewed as a
  polynomial in $Y$ does not vanish on $S$, and if the number of
  distinct roots of $Y\mapsto A(x,Y)$ on $\cc$ is the same for all
  $x$ in $S$, then the
  roots of $A$ are delineable on $S$.
\end{theorem}

Criteria are given in the remainder of Collin's paper for
characterizing the invariance of the number of roots. He uses the
resultant of two polynomials.

Let $A$ and $B$ be polynomials in $\rr[X][Y]$ with $\deg_Y(A)=m$ and
$\deg_Y(B)=n$. The Sylvester matrix of $A$ and $B$ is the $m+n$ by
$m+n$ matrix $M$ whose successive rows contain the coefficients of the
polynomials $Y^{n-1}A(Y),\ldots,YA(Y),A(Y),Y^{m-1}B(Y),\ldots,B(Y)$,
with the coefficient of $Y^i$ occuring in column $m+n-i$. The
polynomial $\rm{Res}(A,B)$, the resultant of $A$ and $B$, is by
definition the determinant of $M$. If the leading coefficient of $A$
vanishes for a particular $x_0$, then the resultant $\rm{Res}(A,A_Y)$
also vanishes at $x_0$ (we recall that $A_Y$ is a shorthand for $\frac{\partial
  A}{\partial Y}$). We note, for subsequent application, that if $A\in
\rr[X,Y]$, $\deg_X(A)\leq d$ and $\deg_Y(A)\leq d$, then
$\rm{Res}(A,A_Y)$ is a polynomial in $\rr[X]$ of degree bounded by
$2d^2-d$.

An immediate corollary of Theorem 1, 2 and 3 in~\cite{Col75} is
\begin{corollary}\label{Cor_criterDelineable}
  Let $A(X,Y)$ be a polynomial in $\rr[X][Y]$. Let $S$ be a connected subset of
  $\rr$. If $\rm{Res} (A,A_Y)(X)$ does not have any
  roots on $S$, then the roots of $A$ are delineable on $S$.
\end{corollary}

Then, in the following, we will consider some subsets of $\rr$ where
the polynomial $\rm{Res} (A,A_Y)$ does not have roots. In particular,
we want this polynomial to be nonzero. We show that this is the case if $A$
is irreducible in $\rr[X,Y]$.

\begin{lemma}\label{Lem_Res}
  Let $A(X,Y)$ be an irreducible polynomial in $\rr[X][Y]$ with
  $\deg_Y(A)\geq 1$. Then
  $\rm{Res} (A,A_Y)$ is not the zero polynomial.
\end{lemma}
\begin{proof}
By Gauss' Lemma, the irreducibility of $A$ in $\rr[X][Y]$ implies that
$A$ is also irreducible in $\rr(X)[Y]$.
Let us suppose that $R(X)=\rm{Res} (A,A_Y)(X) =0$. This implies
that $A$ and $A_Y$ have a common factor 
  $B \in \rr(X)[Y]$ of degree $\deg_Y(B)\geq 1$.
Since $A$ is irreducible in $\rr(X)[Y]$, there exists $C$  in $\rr(X)$
  such that $A=CB$. We thus have $\deg_Y(A)=\deg_Y(B)\leq \deg_Y(A_Y)$.
This is impossible since $\deg_Y(A)\geq 1$.
\end{proof}

%\begin{proof}
%  Let us suppose that $R(X)=\rm{Res} (A,A_Y)(X)$ is zero. It is well known that this condition implies that
%  $A$ and $A_Y$ have a common factor in
%  $\rr(X)[Y]$. Let $B(X,Y)$ be a common factor (we know $\deg_Y(B)\geq 1$). Furthermore, as $A$ is
%  irreducible in $\rr[X][Y]$ and as $\deg_Y(A)\geq 1$, we get that $A$
%  is irreducible in $\rr(X)[Y]$. Hence, there exists $C$ in $\rr(X)$
%  such that $A=CB$, that implies $\deg_Y(A)=\deg_Y(B)\leq \deg_Y(A_Y)$. That is impossible since $A$ is not the zero
%  polynomial.
%\end{proof}

\begin{remark}\label{Rk_A_Y}
  If $(x_0,y_0)$ is a root of $A$ and of $A_Y$ then 
$Y-y_0$ divides the polynomials
  $A(x_0,Y)$ and $A_Y(x_0,Y)$. Hence $\rm{Res}(A,A_Y)(x_0)=0$. Therefore,
  if $\rm{Res}(A,A_Y)$ has no zeros on a subset $S$ of $\rr$, %then
  the system $A(X,Y)=A_Y(X,Y)=0$ does not have solutions on
  $S\times\rr$. This remark will be useful for the proof of
  Lemma~\ref{Lem_T_S} in the next section, and for an application
of the analytic implicit function theorem before Lemma~\ref{rootsonlines}.
\end{remark}

% \section{Real roots}
\section{Intersecting  a sparse curve  with  a  low-degree curve} \label{main}

Recall that a polynomial is said to be $t$-sparse if it has at most
$t$ monomials.
In this section we prove our main result.
\begin{theorem}\label{Cor_genSystem}
  Let $F \in \rr[X,Y]$ be a nonzero bivariate polynomial of degree $d$ and let $G \in \rr[X,Y]$ be a
  bivariate $t$-sparse polynomial.
 %(with $t$ and $d$ positive
  % integers).
The set of real solutions of the system
  \begin{align}\label{Eq_system}
    \begin{cases} F(X,Y)=0 \\ G(X,Y)=0 \end{cases} 
  \end{align} 
has a number of connected components which is $O(d^3t+d^2t^3)$.
\end{theorem}

%We will first consider the case where $F$ is irreducible and the
%system is nondegenerate. 

We will proceed by reduction to the case  where $F$ is irreducible and the
system has finitely many solutions:

\begin{proposition}\label{Thm_irredSystem}
Consider again a nonzero bivariate polynomial $F$ of degree $d$ and 
 a bivariate $t$-sparse polynomial $G$. Assume moreover that $F$ 
is irreducible in $\cc[X,Y]$  and that~(\ref{Eq_system}) 
has finitely many real solutions. 
Then this  system has $O(d^3t+d^2t^3)$ distinct real solutions.
\end{proposition}

We first explain why this proposition implies Theorem~\ref{Cor_genSystem}. 
Let us begin by removing the hypothesis that the system 
has a finite number of solutions.

\begin{corollary}[Corollary of Proposition~\ref{Thm_irredSystem}]\label{Cor_irredSystem}
  Consider again a nonzero bivariate polynomial $F$ of degree $d$ and 
 a bivariate $t$-sparse polynomial $G$. Assume moreover that $F$ 
is irreducible in $\cc[X,Y]$.
  The set of real solutions of~(\ref{Eq_system})  has a number of connected components which is $O(d^3t+d^2t^3)$.
\end{corollary}

\begin{proof}
  There are two cases: 
  \begin{enumerate}
  \item %Either System~(\ref{Eq_system2}) is
    %nondegenerate on $\cc$. 
The system has a finite set of real solutions.
In this case, by Proposition~\ref{Thm_irredSystem}
    there at most  $O(d^3t+d^2t^3)$ solutions.

  \item The set of solutions is infinite.
This implies that $F$ and $G$ share a common factor. Since $F$ is
    irreducible in $\cc$, $F$ must be a factor of $G$. But in this case
    the set of solutions of~(\ref{Eq_system}) is exactly the
    set of zeros  of $F$. By Theorem~\ref{BCSSbezout}, this set has at most 
$d(2d-1)$ connected components.
  \end{enumerate}
\end{proof}

\begin{proof}[Proof of Theorem~\ref{Cor_genSystem} from
  Corollary~\ref{Cor_irredSystem}] 
  By Lemma~\ref{Lem_decompose}, the set of real roots of $F$ is
  the union of the set of real roots of the real irreducible
  factors $F_1,\ldots,F_k$ of $F$ and of a set $\mathcal{U}$ of
  cardinality at most $d^2/4$. Consequently, the number of connected
  components of the set of solutions of~(\ref{Eq_system}) is bounded by the
  sum of the numbers of connected components of the solutions of the
  systems $F_i(x,y)=G(x,y)=0$ for $i\leq k$ and of the system \begin{align*}
    \begin{cases}(X,Y) \in \mathcal{U} \\ G(X,Y)=0. \end{cases}
  \end{align*} 
  The latter system has at most $d^2/4$ solutions.
% (and so $d^2/4$ connected components) 
By Corollary~\ref{Cor_irredSystem}, each
  system $F_i(x,y)=G(x,y)=0$ has at most $O((\deg F_i)^3t+(\deg F_i)^2t^3)$ 
connected components. To conclude, we observe that
  \begin{align*}
    \sum_{i=1}^k \left((\deg F_i)^3t+(\deg F_i)^2t^3\right) & \leq
    \left(\sum_{i=1}^k \deg F_i\right)^3t+\left(\sum_{i=1}^k
      \deg F_i\right)^2t^3 \\
    & \leq d^3t+d^2t^3.
  \end{align*}
\end{proof}

\begin{remark}
  Note that the non-zero condition on $F$ in
  Theorems~\ref{Cor_genSystem} and Proposition~\ref{Thm_irredSystem}
  is important.  Indeed, it is an open problem whether there exists a
  polynomial $P(t)$ which bounds the number of real solutions of any
  system of two $t$-sparse polynomials $G$ and $H$ when this bound is
  finite. However, if we allowed the polynomial $F$ to be $0$ in
  Theorems~\ref{Cor_genSystem}, we would be able to code a system of
  two sparse equations in the system: \begin{align}
    \begin{cases} F=0 \\ G(X,Y)^2+H(X,Y)^2=0. \end{cases}
  \end{align}
\end{remark}

It remains to prove Proposition~\ref{Thm_irredSystem}. In the
following, we will suppose that the system has a finite number of real
solutions. We begin with two basis cases.

\begin{enumerate}
\item If $F(X,Y)=c Y$, then as $G(X,0)\neq 0$ (otherwise $(x,0)$ is a
  solution of~(\ref{Eq_system}) for all $x$ in $\mathbb{R}$), by
  Descartes' rule, the number of roots of the form $(x,0)$ is bounded
  by $2t-1$.

\item If $F_Y(X,Y) = 0$, then $F$ does not depend on $Y$ and there are
  at most $d$ values of $x$ such that $F(x,Y)=0$. For every such value,
  $G(x,Y)$ is a univariate $t$-sparse polynomial so it has at most
  $2t-1$ distinct real roots. Hence, in this case there are at most
  $2td-d$ solutions to~(\ref{Eq_system}).
\end{enumerate}

We have therefore verified the bound of Proposition~\ref{Thm_irredSystem}  in
these two particular cases. 
We will assume in the following we are not in case~1 or~2.

Let us consider the univariate polynomial $\rm{Res}(F,F_Y)$, which is of degree
at most $2d^2-d$ and which is not zero by Lemma~\ref{Lem_Res}. Let $x_1 <\ldots <x_q$ with $q\leq 2d^2-d$ be the real roots
of this polynomial and let $\mathcal{I}=\{(x_i,x_{i+1}) | 0\leq i\leq q\}$ with
$x_0=-\infty$ and $x_{q+1}=+\infty$, be the corresponding set of  open
intervals. We notice that $|\mathcal{I}|\leq 2d^2-d+1$. If $I$ is in $\mathcal{I}$, the roots of $F$ are delineable
on $I$ by Corollary~\ref{Cor_criterDelineable}.

From the definition of delineability, for each interval $I$ in
$\mathcal{I}$, there are $m_I\leq d$ continuous real-valued functions
$\phi_{I,1}<\ldots<\phi_{I,m_I}\ :\ I\rightarrow \rr$ such that 
$F(x,y)=0$ on $I\times\rr$ if and only if there exists $i\leq m_I$ such that $y=\phi_{I,i}(x)$. Moreover, $F_Y(x,\phi_{I,i}(x)) \neq 0$ since
$\rm{Res}(F,F_Y)$ does not vanish on $I$ (see Remark~\ref{Rk_A_Y}).
 The analytic version of the implicit function theorem 
therefore shows that the functions 
$\phi_{I,i}$ are analytic on $I$.

Let us denote $\Omega = \bigcup_{I\in \mathcal{I}} I$. We bound
separately the number $s$ of solutions of system~(\ref{Eq_system}) on
$\Omega\times\rr$ and the number $s^\prime$ of solutions on
$\left(\rr\setminus\Omega\right)\times \rr$. 

\begin{lemma} \label{rootsonlines}
  If $F_Y(X,Y)$ is a non-zero polynomial, the number $s^\prime$ of
  solutions on $\left(\rr\setminus\Omega\right)\times \rr$ of
  System~(\ref{Eq_system}) is at most $2d^3-d^2$.
\end{lemma}

\begin{proof}
  We recall that $\left(\rr\setminus\Omega\right)%\times \rr
= \{x_1,\ldots,x_q\}$ is a finite set of cardinality at most
  $2d^2-d$. For each $i\leq q$, $X-x_i$ does not divide $F$ since $F$
  is irreducible and $F_Y\neq 0$. So the number of roots of $F$ on $\{x_i\}\times\rr$ is finite
  and bounded by $d$. Consequently, $s^\prime \leq 2d^3-d^2$.
\end{proof}

Now, we want to bound the number $s$ of solutions on $\Omega\times\rr$.
To do so, we will bound the number  $s_j^I$ (with $j\leq m_I$)
of solutions of the following system over $I\times\rr$:
\begin{align}\label{Eq_newSystem}
\begin{cases} Y=\phi_j(X) \\ G(X,Y)=0. \end{cases} 
\end{align}
Hence, $\sum_I\sum_{0\leq j\leq m_I} s_j^I=s$ and in particular all
the $s_j^I$ are finite.

The polynomial $G$ is $t$-sparse, so
$G(X,Y)=\sum_{j=1}^ta_jX^{\alpha_j}Y^{\beta_j}$. Then, if $(x,y)$ is a
root of~(\ref{Eq_newSystem}), we have
$G(x,\phi_i(x))=\sum_{j=1}^ta_jx^{\alpha_j}(\phi_i(x))^{\beta_j}=0$.

%For any bivariate polynomial $S$, the function $S(x,\phi_i(x))$
%has an infinity of roots if and only if $F$ divides $S$ (since $F$ is
%irreducible). Thus the condition does not depend on $i$.

Let us assume that there exist real constants $c_1,\ldots,c_t$ (not
all zero) such
that $H(X,Y)=\sum_{j=1}^t c_jX^{\alpha_j}Y^{\beta_j}$ is a multiple of
$F$. In this case, we can consider the polynomial
$\tilde{G}(X,Y)=G-\frac{a_u}{c_u}H$ which is $t-1$ sparse (where $c_u$
is a non-zero coefficient of $H$). Then,
the roots of~(\ref{Eq_system}) are exactly the roots of the
following system:
\begin{align} \label{induction}
  \begin{cases} F(X,Y)=0 \\ \tilde{G}(X,Y)=0.
  \end{cases}  
\end{align}
In this system, the first polynomial has not changed and the number of
terms of the second polynomial has decreased.  We can therefore assume
(by induction on~$t$) that the claimed $O(d^3t+d^2t^3)$ upper bound on
the number of real solutions applies to~(\ref{induction}). We will
therefore assume for the remainder of the proof that if
$H(X,Y)=\sum_{j=1}^t c_jX^{\alpha_j}Y^{\beta_j}$ is a multiple of $F$
then all the constants $c_j$ are zero.

Before stating the next lemma, we recall that $\cal I$ is a finite
list of open intervals defined before Lemma~\ref{rootsonlines} and
that if we want to use Theorem~\ref{thm_Wronskien} for bounding the number of
zeros of $f_1+\ldots+f_t$, we need to bound the number of zeros of
$W(f_1,\ldots,f_s)$ for each $s\leq t$.
\begin{lemma}\label{Lem_T_S}
  For any $s\leq t$, there exists a non-zero polynomial $T_s(X,Y)\in\rr[X,Y]$ of
  degree at most $(1+2d)\binom s2$ in each variable such that for every 
%non trivial
  interval $I$ in $\mathcal{I}$ and every $0\leq i\leq m_I$, 
%we can rewrite the following Wronskian on $I$:
the Wronskian of the $s$ functions 
$x^{\alpha_1}(\phi_i(x))^{\beta_1},\ldots,x^{\alpha_s}(\phi_i(x))^{\beta_s}$
satisfies:
  \begin{align*}
    W(x^{\alpha_1}(\phi_i(x))^{\beta_1},\ldots,x^{\alpha_s}(\phi_i(x))^{\beta_s})
    = \frac{x^{\alpha-\binom s2}\phi_i^{\beta-\binom s2}}{F_Y^{s(s-1)}(x,\phi_i)} T_s(x,\phi_i)
  \end{align*} where $\alpha=\sum_{j=1}^s\alpha_j$ and $\beta=\sum_{j=1}^s\beta_j$.
Moreover, this Wronskian is not identically 0 on $I$.
\end{lemma}

\begin{proof}
  Let $I$ be a an interval in $\mathcal{I}$ and $i$ be an integer
  between $0$ and $m_I$. If $\sum_{j=1}^t
  c_jx^{\alpha_j}\phi_i^{\beta_j}=H(x,\phi_i(x))$ is the zero polynomial, then $F$ divides $H$ 
by irreducibility of $F$.
 It then follows that $H \equiv 0$ by the assumption preceding the lemma. 
The family $x\mapsto
  x^{\alpha_j}(\phi_i(x))^{\beta_j}$  is therefore linearly independent. 
As the functions are analytic on $I$, the Wronskian
  $W(x^{\alpha_1}(\phi_i(x))^{\beta_1},\ldots,x^{\alpha_s}(\phi_i(x))^{\beta_s})$
  is not identically zero.
  By Remark~\ref{Rk_A_Y}, $F_Y(x,\phi_i(x))$ has no zeros on $I$.
  Then using Lemmas~\ref{lem_power} and \ref{Lem_higherDerivative}, 
  \begin{align*}
    \left(x^{\alpha_j}(\phi_i(x))^{\beta_j}\right)^{(p)} & =
    \sum_{k=0}^p \binom pk \left(x^{\alpha_j}\right)^{(k)}
    \left(\phi_i(x)^{\beta_j}\right)^{(p-k)} \\
    & = \sum_{k=0}^p \binom pk \left(x^{\alpha_j}\right)^{(k)}
    \sum_{s\in \mathcal{S}_{p-k}}\left[c_{\beta_j,s}\phi_i^{\beta_j-|s|}\prod_{l=1}^{p-k}\left(\phi_i^{(l)}\right)^{s_l}\right]\\
    & = x^{\alpha_j-p}\phi_i^{\beta_j-p}\sum_{k=0}^p \sum_{s\in
      \mathcal{S}_{p-k}} \left[c^\prime_{\alpha_j,\beta_j,p,s}
      x^{p-k}\phi_i^{p-|s|}\prod_{l=1}^{p-k}\left(\frac{S_l}{F_Y^{2l-1}}\right)^{s_l}\right]\\
    & = \frac{x^{\alpha_j-p}\phi_i^{\beta_j-p}}{F_Y^{2p}}\sum_{k=0}^p
    \sum_{s\in \mathcal{S}_{p-k}}
    \left[c^\prime_{\alpha_j,\beta_j,p,s} x^{p-k}\phi_i^{p-|s|}
      F_Y^{2k+|s|}
      \prod_{l=1}^{p-k} S_l^{s_l}\right]\\
    & =\frac{x^{\alpha_j-p}\phi_i^{\beta_j-p}}{F_Y^{2p}}
    T_{j,p}(x,\phi_i).
  \end{align*} 
  We saw in Lemma~\ref{Lem_higherDerivative} that $\deg(S_l) \leq
  2l-1$. As a result, $T_{j,p}(X,Y)$ is a
  polynomial of degree in $X$ bounded by
\begin{align*}
  \deg_X(T_{j,p}) & \leq \max_{k,s}\left(p-k + (2k+|s|)d+\sum_{l=1}^{p-k} s_l(2l-1)d\right)
  \\
  & \leq \max_{k,s}\left(p-k+2kd+|s|d+2d(p-k)-d|s|\right) \\
  & \leq 2dp+p
\end{align*} and of degree in $Y$ bounded by
\begin{align*}
  \deg_Y(T_{j,p}) & \leq
  \max_{k,s}\left(p-|s|+(2k+|s|)(d-1)+\sum_{l=1}^{p-k}s_l(2l-1)d\right) \\
  & \leq \max_{k,s}\left(p-|s|+2kd-2k+|s|d-|s|+2dp-2dk-d|s|\right) \\
  & \leq \max_{k,s}\left(p-2|s|-2k+2dp\right) \\
  & \leq p+2dp.
\end{align*}
Moreover $T_{j,p}$ does not depend on $\phi_i$ by Lemma~\ref{Lem_higherDerivative}.
Hence, the Wronskian is a bivariate rational function:
\begin{align*}
  W\left(x^{\alpha_1}(\phi_i(x))^{\beta_1},\ldots,x^{\alpha_s}(\phi_i(x))^{\beta_s}\right)
  =\frac{x^{\alpha-\binom s2}\phi_i^{\beta-\binom s2}}{F_Y^{s(s-1)}(x,\phi_i)} T_s(x,\phi_i)
\end{align*} where $\alpha=\sum_{j=1}^s\alpha_j$,
$\beta=\sum_{j=1}^s\beta_j$ and $T_s(X,Y)$ is a polynomial of degree
bounded by $(1+2d)\binom s2$ in each variable, which does not depend
on $I$ and $i$.
\end{proof}

Let us count the number $v_{I,i}$ of roots of $\phi_i$ on $I$. We assumed
that $Y$ does not divide $F$, so the univariate polynomial $F(X,0)$ is
not zero and is of degree at most $d$. If  $v_I$ denotes  the number of roots
of $F(X,0)$ on $I$, we have
 $\left(\sum_{I\in\mathcal{I}}v_I\right)\leq d$. 
Since each root of $\phi_i$ is by definition a root of $F(X,0)$,
this implies  that $\phi_i$ has at most $v_I$ roots on~$I$.% $\rr$.

For any $I$ and $i$, let us count now the number $r_{I,i}^s$ of roots of
$T_s(x,\phi_i(x))$. This number is finite since the Wronskian of 
Lemma~\ref{Lem_T_S} would otherwise be identically 0. 
Furthermore, let us denote by $r^s$ the number of solutions on 
$\Omega\times\rr$ of the system
\begin{align} \label{Syst_TF}
  \begin{cases}
    T_s\left(X,Y\right)=0 \\
    F\left(X,Y\right)=0.
  \end{cases}
\end{align}
Thus, $r^s=\left(\sum_I\sum_ir_{I,i}^s\right)$ is finite.

Finally, by Lemma~\ref{Lem_realBezout} and as the total degree of
$T_s$ is bounded by $2(1+2d)\binom s2$, the number $r^s$ of roots of (\ref{Syst_TF}) is
bounded by $d^2/4+2d(1+2d)\binom s2$.

Then, by Lemma~\ref{Lem_T_S}, for any $I$ and $i$,
$W\left(x^{\alpha_1}(\phi_i(x))^{\beta_1},\ldots,x^{\alpha_s}(\phi_i(x))^{\beta_s}\right)$
has at most $\mathbbm{1}_I(0)+v_I+r_{I,i}^s$ real roots and %by 
Theorem~\ref{thm_Wronskien} shows that the number $s_{I,i}$ 
of distinct real roots of
$G(x,\phi_i(x))$ is bounded by 
\begin{align*}
  t-1+2\sum_{s=1}^t\left(\mathbbm{1}_I(0)+v_I+r_{I,i}^s\right) = t-1 +2t \mathbbm{1}_I(0)+2tv_I +2\sum_{s=1}^tr_{I,i}^s. 
\end{align*}

Hence, as $|\mathcal{I}|\leq 2d^2-d+1$, as $m_I\leq d$ for each interval
$I$ in $\mathcal{I}$, and as $\left(\sum_{I\in\mathcal{I}}v_I\right)\leq
d$,
\begin{align*}
  & s=\sum_{I\in\mathcal{I}}\sum_{i\leq m_I}s_{I,i} \\
  & \leq (2d^2-d+1)d(t-1)+2dt+2td^2+2\sum_{s=1}^tr^s \\
  & \leq (2d^2-d+1)d(t-1)+2dt+2td^2+2\sum_{s=1}^t \left[d^2/4+2(1+2d)\binom s2 d\right] \\
  & =(2d^3t+4d^2t^3)(1+o(1)).
\end{align*}

%Hence for each $i$, System~(\ref{Eq_newSystem}) has at most
%$B(t,d)=2t^3d^2+t^3d-4d^2t^2-2t^2d+5/2td^2+3td+3t-1/2d^2-2d-3$ real roots on $I$ and it follows that the number of
%roots of System~(\ref{Eq_system2}) is bounded by
%$d(d-1)+(d(d-1)+1)dB(t,d)= O(t^3d^5)$. 
This completes the proof of Proposition~\ref{Thm_irredSystem}, and of the main 
theorem.

\section*{Acknowledgments}

We thank the two referees for suggesting several improvements in the presentation of the paper.

%\bibliographystyle{plain}
%\bibliography{Real_Roots_System}

\end{document}